\documentclass[12pt,amsfonts]{article}
\usepackage[utf8]{inputenc}
\usepackage[T1]{fontenc}
\usepackage{tikz}
\usepackage[english]{babel}
\usepackage{mathrsfs}
\usepackage{amssymb}
\usepackage[cmex10]{amsmath}
\usepackage{amsthm}
\usepackage{wrapfig}
\usepackage{latexsym}
\usepackage{amsfonts}
\usepackage{color}
\usepackage[font=small, labelfont=bf]{caption}
\usepackage{ctable}
\usepackage{floatrow}
\usepackage[colorlinks=true, citecolor=red, linkcolor=blue]{hyperref}
\usepackage{graphicx}
\usepackage{subfig}
\usepackage{verbatim}
\usepackage{epstopdf}
\usepackage{float}
\numberwithin{equation}{section}
\newtheorem{theo}{Theorem}[section]

\newtheorem{prop}[theo]{Proposition}

\theoremstyle{remark}

\newcommand{\be}{\begin{eqnarray*}}
    \newcommand{\ee}{\end{eqnarray*}}
\newcommand{\ben}{\begin{eqnarray}}
\newcommand{\een}{\end{eqnarray}}

\newcommand{\lb}[1]{\left[\begin{array}{#1}}
    \newcommand{\rb}{\end{array}\right]}
\newcommand{\lp}[1]{\left(\begin{array}{#1}}
    \newcommand{\rp}{\end{array}\right)}

\newcommand{\leftd}[1]{\left\{\begin{array}{#1}}
    \newcommand{\rightd}{\end{array}\right.}

\def\B {\mathbf{B}}

\def\M {\mathbf{M}}

\def\S {\mathbf{S}}

\def\b {\mathbf{b}}

\def\p {\mathbf{p}}

\def\s {\mathbf{s}}

\def\Rb {\mathbb{R}}

\usepackage[title,toc,page,header]{appendix}

\begin{document}

\renewcommand{\thefootnote}{\arabic{footnote}}

\begin{center}
{\Large \textbf{Parametrizations, weights, and optimal prediction: Part 1}} \\[0pt]
~\\[0pt]
Azzouz Dermoune\footnote{
 Laboratoire Paul Painlev\'e,
USTL-UMR-CNRS 8524. UFR de Math\'ematiques, B\^at. M2. 59655
Villeneuve d'Ascq C\'edex, France. Email:
\texttt{azzouz.dermoune@univ-lille1.fr}}, Khalifa
Es-Sebaiy\footnote{ Cadi Ayyad University, Marrakesh, Morocco.
E-mail: \texttt{k.essebaiy@uca.ma}}*,  Mohammed Es.Sebaiy\footnote{
Cadi Ayyad University, Marrakesh, Morocco. E-mail:
\texttt{mohammedsebaiy@gmail.com}} and Jabrane Moustaaid\footnote{ Cadi Ayyad University, Marrakesh,
 Morocco. E-mail: \texttt{jabrane.mst@gmail.com}\\ * Corresponding author}
\\[0pt]
\textit{  Lille University   and Cadi Ayyad University }\\[0pt]
~\\[0pt]
\end{center}

\begin{abstract}
We consider the problem of the annual mean temperature prediction.
The years taken into account and the corresponding annual mean
temperatures are denoted by $0,\hdots, n$ and $t_0$, $\hdots$,
$t_n$, respectively. We propose to predict the temperature $t_{n+1}$
using the data $t_0$, $\hdots$, $t_n$. For each $0\leq l\leq n$ and
each parametrization $\Theta^{(l)}$ of the  Euclidean space
$\Rb^{l+1}$ we construct a list of weights for the data
$\{t_0,\hdots, t_l\}$ based on the rows of $\Theta^{(l)}$ which are
correlated with the constant trend. Using these weights we define a
list of predictors of $t_{l+1}$ from the data $t_0$, $\hdots$,
$t_l$. We analyse how the parametrization affects the prediction,
and provide three optimality criteria for the selection of weights
and parametrization. We illustrate our results for the annual mean
temperature of France and Morocco.
\end{abstract}


\section{Motivation}
We consider the problem of the annual mean temperature prediction.
The years taken into account and the corresponding annual mean temperatures are denoted
by $0,\hdots, n$ and $t_0$, $\hdots$, $t_n$, respectively.
We model the behavior of the temperature $i\to t_i:=s(i)$ by
the column vector $\s^{(n)}=(s(0), \ldots, s(n))^\top\in\Rb^{n+1}$.
The aim is to predict the temperature $s(n+1)$ at the year $n+1$.

For each $0\leq l\leq n$ and each parametrization
$\Theta^{(l)}=(\theta_{ji}^{(l)}:\quad i,j=0, \hdots, l)$ of the
Euclidean space $\Rb^{l+1}$ we construct a list of weights for the
data $\{t_0,\hdots, t_l\}$ based on the row $\theta_j^{(l)}$ of
$\Theta^{(l)}$ which is correlated with the constant trend $\{{\bf
1}^{(l)}\}^\top=(1, \hdots, 1)$, i.e., $\theta_j^{(l)}{\bf
1}^{(l)}\neq 0$. We analyze how the parametrization $\Theta^{(l)}$
affects the prediction. We also propose a list of criteria for
selecting optimal parametrization and weights. We illustrate our
results for the annual mean temperature of France and Morocco.

The present paper is the first part of a list of works in preparation.
These works are directly related to \cite{DP1}, \cite{DP2} and \cite{DRW},
see also \cite{DDR1}, \cite{DDR2}, \cite{DW}.

\section{Parametrization}
Let $0\leq l\leq n$ be an integer and $\Theta^{(l)}$ be any invertible $(l+1)\times(l+1)$ real matrix.
Its $j$-th row is denoted by $\theta_j^{(l)}$ and then its entry $(j,i)$
is equal to $\theta_{ji}^{(l)}$, with $i,j=0, \ldots, l$.
Its inverse $\{\Theta^{(l)}\}^{-1}$ is denoted by $\B^{(l)}$. The $j$-th column of
$\B^{(l)}$ is denoted by $\b_j^{(l)}$ and then its entry $(i,j)$ is equal to
$b_{ij}^{(l)}$. Let $\s^{(l)}=(s(0), \ldots, s(l))^\top\in \Rb^{l+1}$ be any column vector. The equality
\be
\s^{(l)}=\B^{(l)}\Theta^{(l)}\s^{(l)}
\ee
tells us that
\ben
\s^{(l)}=\sum_{j=0}^l\theta_j^{(l)}\s^{(l)}\,\b_j^{(l)}.
\label{parametrization}
\een
Hence the columns $[\b_0^{(l)}, \ldots, \b_l^{(l)}]$ of the matrix $\B^{(l)}$ form
a basis of $\Rb^{l+1}$, and $(\theta_0^{(l)}\s^{(l)}, \ldots,\theta_l^{(l)}\s^{(l)})$ are the coordinates
of the vector $\s^{(l)}$ in the basis $\B^{(l)}$.

\begin{prop} \label{Il}Let ${\bf 1}^{(l)}=(1, \hdots, 1)^\top\in\Rb^{l+1}$
denotes the constant trend written as column vector, and $I(l)=\{j:\quad \theta_j^{(l)}{\bf 1}^{(l)}\neq 0\}$.
We have for $i=0, \ldots, l$,
\be
\sum_{j\in I(l)}\theta_{j}^{(l)}{\bf 1}^{(l)}b_{ij}^{(l)}=1,\\
s(i)=\sum_{j\in I(l)}\theta_j^{(l)}\s^{(l)}\,b_{ij}^{(l)}+\sum_{j\notin I(l)}\theta_j^{(l)}\s^{(l)}\,b_{ij}^{(l)}.
\ee
\end{prop}
\begin{proof} It is the consequence of the equality
\be
1=\sum_{j=0}^l\theta_j^{(l)}{\bf 1}^{(l)}\,b_{ij}^{(l)}\\
=\sum_{j\in I(l)}\theta_j^{(l)}{\bf 1}^{(l)}\,b_{ij}^{(l)},
\ee
with $i=0, \hdots, l$.
\end{proof}
If $s(i)$ oscillates around some constant $c$, then $\sum_{j\notin I(l)}\theta_j^{(l)}\s^{(l)}\,b_{ij}^{(l)}$
oscillates around 0. Hence the component $\sum_{j\in I(l)}\theta_j^{(l)}\s^{(l)}\,b_{ij}^{(l)}$ seems to be
the bulk component of $s(i)$, and  $\sum_{j\notin I(l)}\theta_j^{(l)}\s^{(l)}\,b_{ij}^{(l)}$
its residual component. Roughly speaking, the most important coordinates are those correlated with the constant trend
${\bf 1}^{(l)}$, i.e., $(\theta_j^{(l)}:\quad j\in I(l))$.

\section{Conservative rows and selection criteria}
\subsection{Conservative rows}
The row $\p=(p_0, \ldots, p_l)\in\Rb^{l+1}$ is conservative
if
\be
\sum_{i=0}^l p_i=1.
\ee
If $p_i\geq 0$ for all $i$, then $\p$ is a probability distribution on the set $\{0,\hdots, l\}$.
The set of conservative rows is denoted by
\be
M_1(\{0, \hdots, l\})=\{\p\in \Rb^{l+1}:\quad \sum_{i=0}^lp_i=1\}.
\ee
The mean and the variance of $\s^{(l)}$ w.r.t. to $\p$ are defined respectively by
\be
m_1(\p)=\p\s^{(l)}=\sum_{i=0}^lp_is(i),\\
\mbox{var}(\p)=\sum_{i=0}^lp_i|s(i)-m_1(\p)|^2.
\ee
We have the famous equality
\be
\sum_{i=0}^lp_ix_i^2=m_1^2(\p)+\mbox{var}(\p).
\ee
Observe also as in the probabilistic case, the minimizer
\be
\arg\min\{\sum_{i=0}^lp_i|s(i)-a|^2:\quad a\in\Rb\}=m_1(\p),
\ee
and the error
\be
\min\{\sum_{i=0}^lp_i|s(i)-a|^2:\quad a\in\Rb\}=\mbox{var}(\p).
\ee
\subsection{Selection criterion}
The set $\mathcal{P}$ contains a finite number of parametrizations of the Euclidean spaces $(\Rb^{2}, \hdots, \Rb^{n+1})$. An element of $\mathcal{P}$ is
a parametrization $\Theta:=(\Theta^{(1)}, \hdots, \Theta^{(n)})$
of the Euclidean spaces $(\Rb^{2}, \hdots, \Rb^{n+1})$.

Let us give for each $1\leq l\leq n$ and each parametrization $\Theta^{(l)}$
a finite subset $W(\Theta^{(l)})$ of the set of conservative rows
$M_1(\{0, \hdots, l\})$. We get the subset
$W(\Theta):=W(\Theta^{(1)})\times \hdots\times W(\Theta^{(n)})$
of $M_1(\{0,1\})\times \hdots \times M_1(\{0, \hdots, n\})$.

A selection criterion $S$ picks a unique element
\be
SW(\Theta)=(SW(\Theta^{(1)}), \hdots, SW(\Theta^{(n)}))\in M_1(\{0,1\})\times \hdots \times M_1(\{0, \hdots, n\})
\ee
from the set $W(\Theta)$.

\subsection{Prediction cost}
We propose for each $l=1, \hdots, n$,
\be
SW(\Theta^{(l)})\s^{(l)}
\ee
as a prediction of $s(l+1)$. The cost of these predictors for $q=1,2,+\infty$, and $L<n$ fixed,
is measured by
\be
&&\mbox{Cost}(SW(\Theta),q)=\frac{\sum_{l=L}^{n-1}|s(l+1)-SW(\Theta^{(l)})\s^{(l)}|^q}{n-L},\quad q=1,2,\\
&&:=\max\{|s(l+1)-SW(\Theta^{(l)})\s^{(l)}|:\quad L\leq l\leq n-1\},\quad q=+\infty.
\ee
Let $\mathcal{S}$ be a finite set of selection criteria.
The optimal selection criterion $S^*W(\Theta^*,q)$
is the minimizer
\be
S^*W(\Theta^*,q)=\arg\min\{\mbox{Cost}(SW(\Theta),q):\quad \Theta\in \mathcal{P},\quad S\in\mathcal{S}\}.
\ee
In this work we consider the sets \be
W(\Theta^{(l)})=\left\{\frac{\theta_j^{(l)}}{\theta_j^{(l)}{\bf
1}^{(l)}}:\,\, 0\leq j\leq l,\,\, \theta_j^{(l)}{\bf 1}^{(l)}\neq
0\right\}, \ee with $l=1, \hdots, n$. We recall that $I(l)=\{0\leq
j\leq l:\quad \theta_j^{(l)}{\bf 1}^{(l)}\neq 0\}$. Observe that for
each parametrization $\Theta^{(l)}$, the set $I(l)$ is not empty.

For simplicity we denote for each selection criterion $S$
\be
S(\Theta^{(l)}):=SW(\Theta^{(l)}),\quad l=1, \hdots, n.
\ee
Now, we are going to define our selection criteria.

\subsection{The selection criterion $S_u$}
Let $I(l)=\{j(0), \hdots, j(card(I(l))-1)\}$ be the elements of the set $I(l)$
with $j(0)< \hdots < j(card(I(l))-1)$. We define for each $u\leq n$ fixed the selection criterion
\be
S_u(\Theta^{(l)})=\frac{\theta_{j(\min(u,card(I(l))-1))}^{(l)}}{\theta_{j(\min(u,card(I(l))-1))}^{(l)}{\bf 1}^{(l)}}.
\ee
If $u=0$, then
\be
S_0(\Theta^{(l)})=\frac{\theta_{j(0)}^{(l)}}{\theta_{j(0)}^{(l)}{\bf 1}^{(l)}}.
\ee
If $u=n$, then
\be
S_n(\Theta^{(l)})=\frac{\theta_{j(card(I(l))-1)}^{(l)}}{\theta_{j(card(I(l))-1)}^{(l)}{\bf 1}^{(l)}}.
\ee
\subsection{The winning conservative rows $S_{u(q,\Theta)}(\Theta)$}
Given $\Theta$ and $q=1,2, +\infty$, the optimal selection criterion among $(S_u:\quad u=0, \hdots,n)$ is given by the minimizer
\be
u(q,\Theta)=\arg\min\{\mbox{Cost}(S_u(\Theta),q):\quad u=0, \hdots, n\}.
\ee
Hence $S_{u(q,\Theta)}(\Theta)$ is the optimal conservative rows among the set
$(S_u(\Theta):\quad u=0, \hdots,n)$ of conservative rows.

\subsection{The selection criterion $S_{\mbox{\small{mean}}}$}
For $1\leq l\leq n$, we consider the selection criterion
\be
S_{\mbox{\small{mean}}}(\Theta^{(l)})=\sum_{j\in I(l)}\frac{\theta_{j}^{(l)}}{card(I(l))\theta_j^{(l)}{\bf 1}^{(l)}}.
\ee
Observe that for the canonical parametrization $\Theta^{(l)}=(\delta_i^j: 0\leq i,j\leq l)$, we have
\be
S_{\small{\mbox{\small{mean}}}}(\Theta^{(l)})=(\frac{1}{l+1},\hdots,\frac{1}{l+1}).
\ee
\subsection{The winning conservative rows $S_{1,q}(\Theta)$}
Given $\Theta$ and $q=1,2, +\infty$, the optimal conservative rows among
$S_{u(q,\Theta)}(\Theta)$ and $S_{\mbox{\small{mean}}}(\Theta)$
is the minimizer
\be
S_{1,q}(\Theta):=\arg\min\{\mbox{Cost}(S_{u(q,\Theta)}(\Theta),q),\mbox{Cost}(
S_{\mbox{\small{mean}}}(\Theta),q)\}.
\ee
\subsection{The selection criterion $S_{u\mbox{\small{tail1}}}$}
For a fixed $u\leq n$ the set
\be
J_1(l,u)=\arg\max\{\sum_{i=\min(u,l)}^l\frac{\theta_{ji}^{(l)}}{\theta_j^{(l)}{\bf 1}^{(l)}}:\quad j\in I(l)\}
\ee
may be not a singleton. It
furnishes the selection criterion
\be
S_{\mbox{\small{tail1}}}(\Theta^{(l)})=\sum_{j\in J_1(l,u)}\frac{\theta_{j}^{(l)}}{card(J_1(l,u))\theta_{j}^{(l)}{\bf 1}^{(l)}}.
\ee
As a simple example, if $\Theta^{(l)}$ is the canonical parametrization and $u=n$, then
\be
J_1(l,n)=l.
\ee
If $u=n-1$, then $J_1(l,n-1)=l$ for $l\leq n-1$, and $J_1(n,n-1)=\{n-1,n\}$.
\subsection{The winning conservative rows $S_{u(q,\Theta)\mbox{\small{tail1}}}(\Theta)$}
The optimal selection criterion among $(S_{u\mbox{\small{tail1}}}:\quad u=0, \hdots,n)$ is the minimizer of
\be
S_{u(q,\Theta)\mbox{\small{tail1}}}=\arg\min\{\mbox{Cost}(\Theta,S_{u\mbox{\small{tail1}}}(\Theta),q):\quad u=0, \hdots, n\}.
\ee
\subsection{The winning conservative rows $S_{2,q}(\Theta)$}
The winner for each $q=1,2, +\infty$ fixed,
among $S_{u(q,\Theta)\mbox{\small{tail1}}}(\Theta)$ and $S_{1,q}(\Theta)$
is the minimizer
\be
S_{2,q}(\Theta):=\arg\min\{\mbox{Cost}(S_{u(q,\Theta)\mbox{\small{tail1}}}(\Theta),q),\mbox{Cost}(S_{1,q}(\Theta),q)\}.
\ee
\subsection{The selection criterion $S_{u\mbox{\small{tail2}}}$}
For a fixed $u\leq n$ the set
\be
J_2(l,u)=\arg\max\{\max\{\frac{\theta_{ji}^{(l)}}{\theta_j^{(l)}{\bf 1}^{(l)}}:\quad \min(u,l)\leq i\leq l\}:\quad j\in I(l)\},
\ee
may be not a singleton. It
furnishes the selection criterion
\be
S_{\mbox{\small{tail2}}}(\Theta^{(l)})=\sum_{j\in J_2(l,u)}\frac{\theta_{j}^{(l)}}{card(J_2(l,u))\theta_{j}^{(l)}{\bf 1}^{(l)}}.
\ee
\subsection{The winning conservative rows $S_{u(q,\Theta)\mbox{\small{tail2}}}(\Theta)$}
The optimal selection criterion among $(S_{u\mbox{\small{tail2}}}:\quad u=0, \hdots,n)$ is the minimizer of
\be
S_{u(q,\Theta)\mbox{\small{tail2}}}=\arg\min\{\mbox{Cost}(S_{u\mbox{\small{tail2}}}(\Theta),q):\quad u=0, \hdots, n\}.
\ee
\subsection{The winning conservative rows $S_{3,q}(\Theta)$}
The winner for each $q=1,2, +\infty$ fixed,
among $S_{u(q,\Theta)\mbox{\small{tail2}}}(\Theta)$ and $S_{2,q}(\Theta)$
is the minimizer
\be
S_{3,q}(\Theta):=\arg\min\{\mbox{Cost}(S_{u(q,\Theta)\mbox{\small{tail2}}}(\Theta),q),\mbox{Cost}(S_{2,q}(\Theta),q)\}.
\ee
\subsection{The selection criterion $S_{\small{maxcor}}$}
The set
\be
J_3(l)=\arg\max\{\frac{|\theta_j^{(l)}{\bf 1}^{(l)}|}{\sqrt{l+1}\|\theta_j^{(l)}\|}:\quad j\in I(l)\}
\ee
of the coordinates $\theta_j^{(l)}$ highly correlated with the constant trend ${\bf 1}^{(l)}$,
furnishes the selection criterion
\be
S_{\mbox{\small{maxcor}}}(\Theta^{(l)})=\sum_{j\in J_3(l)}\frac{\theta_{j}^{(l)}}{card(J_3(l))\theta_{j}^{(l)}{\bf 1}^{(l)}}.
\ee
\subsection{The winning conservative rows $S_{4,q}(\Theta)$}
The winner for each $q=1,2, +\infty$ fixed,
among $S_{\mbox{maxcor}}(\Theta)$ and $S_{3,q}(\Theta)$
is the minimizer
\be
S_{4,q}(\Theta):=\arg\min\{\mbox{Cost}(S_{\mbox{maxcor}}(\Theta),q),\mbox{Cost}(S_{3,q}(\Theta),q)\}.
\ee
\subsection{The selection criterion $S_{q\mbox{\small{nearU}}}$}
We consider the set
\be
J_4(l,q)=\arg\min\{\|\frac{\theta_j^{(l)}}{\theta_j^{(1)}{\bf 1}^{(l)}}-(\frac{1}{l+1}, \hdots,\frac{1}{l+1})\|_q:\quad j\in I(l)\}
\ee
of the nearest conservative rows $\frac{\theta_j^{(l)}}{\theta_j^{(1)}{\bf 1}^{(l)}}$ to the uniform
conservative row $(\frac{1}{l+1},\hdots, \frac{1}{l+1})$,
and the corresponding selection criterion
\be
S_{q\mbox{\small{nearU}}}(\Theta^{(l)})=\sum_{j\in J_4(l,q)}\frac{\theta_{j}^{(l)}}{card(J_4(l,q))\theta_{j}^{(l)}{\bf 1}^{(l)}}.
\ee
Here $\|\cdot\|_q$ denotes the $l(q)$-norm with $q=1,2,+\infty$.
\subsection{The winning conservative rows $S_{q(\Theta)\mbox{\small{nearU}}}(\Theta)$}
For each $q=1,2,+\infty$ fixed let us denote by
\be
q(\Theta)=\arg\min\{\mbox{Cost}(S_{q_1\mbox{\small{nearU}}}(\Theta),q):\quad q_1=1,2,+\infty\},
\ee
and then we obtain the winning conservative rows $S_{q(\Theta)\mbox{\small{nearU}}}(\Theta)$ among
the three conservative rows $(S_{q_1\mbox{\small{nearU}}}(\Theta):\quad q_1=1,2,+\infty)$.

\subsection{The winning conservative rows $S_{5,q}(\Theta)$}
For each $q=1,2, +\infty$ fixed, the winner
among $S_{q(\Theta)\mbox{\small{nearU}}}(\Theta)$ and $S_{4,q}(\Theta)$
is the minimizer
\be
S_{5,q}(\Theta):=\arg\min\{\mbox{Cost}(S_{q(\Theta)\mbox{\small{nearU}}}(\Theta),q),\mbox{Cost}(S_{4,q}(\Theta),q)\}.
\ee
\subsection{The selection criterion $S_{u\mbox{\small{var}}}$}
For each $j\in I(l)$ the variance of the data $\s^{(l)}$ w.r.t. to the conservative row
$\frac{\theta_j^{(l)}}{\theta_j^{(l)}{\bf 1}^{(l)}}$ is denoted by $var(l,j)$.
We define the one-to-one map $\sigma^{(l)}$
from $\{0,\hdots, card(I(l))-1\}$ to $I(l)$ as follows.
The integer $\sigma^{(l)}(0)$ is the first element of
\be
\arg\min\{\mbox{var}(l,j):\quad j\in I(l)\}.
\ee
By induction for $k<card(I(l))-1$ the integer $\sigma^{(l)}(k+1)$
is the first element of
\be
\arg\min\{\mbox{var}(l,j):\quad j\in I(l)\setminus\{\sigma^{(l)}(0), \hdots, \sigma^{(l)}(k)\}\}.
\ee
We define for a fixed $u\leq n$ the index
\be
j(l,u)=\sigma^{(l)}(\min(u,card(I(l))-1)),
\ee
and the selection criterion
\be
S_{u\mbox{\small{var}}}(\Theta^{(l)})=\frac{\theta_{j(l,u)}^{(l)}}{\theta_{j(l,u)}^{(l)}{\bf 1}^{(l)}}.
\ee
If $u=n$, then $j(l,n)=\sigma^{(l)}(card(I(l))-1)$ is the index of the largest variance.
If $u=0$, then $j(l,0)=\sigma^{(l)}(0)$ is the index of the smallest variance.

\subsection{The winning conservative rows $S_{u(q,\Theta)\mbox{\small{var}}}(\Theta)$}
The optimal selection criterion among $(S_{u\mbox{\small{var}}}:\quad u=0, \hdots,n)$ is the minimizer of
\be
S_{u(q,\Theta)\mbox{\small{var}}}=\arg\min\{\mbox{Cost}(\Theta,S_{u\mbox{\small{var}}}(\Theta),q):\quad u=0, \hdots, n\}.
\ee
\subsection{The winning  conservative sequence $S_{6,q}(\Theta)$}
The winner for each $q=1,2, +\infty$ fixed,
among $S_{u(q,\Theta)\mbox{\small{var}}}(\Theta)$ and $S_{5,q}(\Theta)$
is the minimizer
\be
S_{6,q}(\Theta):=\arg\min\{\mbox{Cost}(\Theta,S_{u(q,\Theta)\mbox{\small{var}}}(\Theta),q),\mbox{Cost}(\Theta,S_{5,q}(\Theta),q)\}.
\ee
\subsection{The selection criterion $S_{uv\small{fd}}$}
We define for each fixed $\tilde{l}\leq l$
the permutation $\sigma^{(l,\tilde{l})}$ of the set $I(l)$ as follows.
The integer $\sigma^{(l,\tilde{l})}(0)$ is the first element of
\be
\arg\min\{|\frac{\theta_j^{(l)}\s^{(l)}}{\theta_j^{(l)}{\bf 1}^{(l)}}-s(\tilde{l})|:\quad j\in I(l)\}.
\ee
By induction for $k< card(I(l))-1$, $\sigma^{(l,\tilde{l})}(k+1)$ is the first element of
\be
\arg\min\{|\frac{\theta_j^{(l)}\s^{(l)}}{\theta_j^{(l)}{\bf 1}^{(l)}}-s(\tilde{l})|:\quad j\in I(l)\setminus\{\sigma^{(l,\tilde{l})}(0), \hdots, \sigma^{(l,\tilde{l})}(k)\}\}.
\ee
Let $0\leq u\leq n$, $0\leq v\leq n$ fixed. The selection criterion
\be
j(l,u,v)=\sigma^{(l,\min(l,u))}(\min(v,card(I(l))-1))
\ee
furnishes the selection criterion
\be
S_{uvfd}(\Theta)=\frac{\theta_{j(l,u,v)}^{(l)}}{\theta_{j(l,u,v)}^{(l)}{\bf 1}^{(l)}}.
\ee
If $u=v=n$, then $j(l,n,n)=\sigma^{(l,l)}(card(I(l))-1)$ is the index of the farest element
$\frac{\theta_j^{(l)}\s^{(l)}}{\theta_j^{(l)}{\bf 1}^{(l)}}$ from $s(l)$.
If $u=0, v=n$, then $j(l,0,n)=\sigma^{(l,0)}(card(I(l))-1)$ is the index of the farest
element $\frac{\theta_j^{(l)}\s^{(l)}}{\theta_j^{(l)}{\bf 1}^{(l)}}$
from $s(0)$.
If $u=v=0$, then $j(l,0,0)=\sigma^{(l,0)}(0)$ is the index of the nearest
element $\frac{\theta_j^{(l)}\s^{(l)}}{\theta_j^{(l)}{\bf 1}^{(l)}}$ from $s(0)$.

\subsection{The winning conservative rows $S_{u(q,\Theta)v(q,\Theta)\small{fd}}(\Theta)$}
Given $\Theta$ and $q=1,2,+\infty$, the optimal conservatrice sequence $S_{u(q,\Theta)v(q,\Theta)\small{fd}}(\Theta)$ among $(S_{uv\small{fd}}(\Theta):\quad u,v=0, \hdots,n)$ is
given by the minimizer
\be
(u(q,\Theta),v(q,\Theta))=\arg\min\{\mbox{Cost}(S_{uv\small{fd}}(\Theta),q):\quad u,v=0, \hdots, n\}.
\ee

\subsection{The winning conservative rows $S_{7,q}(\Theta)$}
Given $\Theta$ and $q=1,2,+\infty$, the optimal conservative rows
$S_{7,q}(\Theta)$
among $S_{u(q,\Theta)v(q,\Theta)\mbox{fd}}(\Theta)$ and $S_{6,q}(\Theta)$
is the minimizer
\be
\arg\min\{\mbox{Cost}(S_{u(q,\Theta)v(q,\Theta)\mbox{fd}}(\Theta),q,L),\mbox{Cost}(S_{6,q}(\Theta),q)\}.
\ee

\section{The winning conservative rows $S_{7}(\Theta(q))$}
We constructed for $q=1,2,+\infty$ fixed and each parametrization $\Theta$ the optimal conservative rows
$S_{7,q}(\Theta)$.
Assume that we have a finite set $\mathcal{P}$ of parametrizations. The minimizer $\Theta^*(q)$
of the map
\be
\Theta\to \mbox{Cost}(S_{7,q}(\Theta),q)
\ee
furnishes the optimal selection criterion $S_7(\Theta(q)):=S_{7,q}(\Theta(q))$.

\section{Application to parametrizations given by the energy of the spline}
We identify for the integer $l\geq 1$ the space $\Rb^{l+1}$ with the space of the natural cubic splines
$S_{3,nat}(0,\ldots, l)$ having the knots $0, \ldots, l$.
Let us denote $S_{3}(0,\ldots, l)$ the set of cubic splines
having the knots $0, \ldots, l$.
We recall that an element $s\in S_3$ is a $C^2$ map
on $[0,l]$ and
is a polynomial of degree three on each interval $[i,i+1)$ for
$i=0$,\ldots, $l-1$.
More precisely, let
\be
p_0 =s(0), \ldots, p_l=s(l),\quad q_0 =s'(0), \ldots, q_{l} =s'(l),\\
u_0 =s''(0), \ldots, u_{l} =s''(l),\quad v_0 = s'''(0+), \ldots, v_{l-1}=s'''(l-1+)
\ee
be respectively the values of $s$ and its derivatives up to order three on the knots.
We have for $i=0,\ldots, l-1$,
\be
s(t)=p_i+q_i(t-i)+(t-i)^2u_i/2+(t-i)^3v_i/6,\quad t\in [i,i+1).
\ee
The following constraint guarantees the hypothesis that $s$ is $C^2$:
\ben
p_i+q_i+u_i/2+v_i/6=p_{i+1},\quad \label{c0}\\
q_i+u_i+v_i/2=q_{i+1},\quad \label{c1}\\
v_i=s^{(3)}(t_i)=u_{i+1}-u_i.\quad \label{c2}
\een
It is well known \cite{deBoor} that $S_3(0,\ldots, l)$ has the dimension $l+3$, see also \cite{Craven} and \cite{Wahba1990}.
Hence an element $s\in S_3(0,\ldots, l)$ is completely defined by
$l+3$ independent parameters. Moreover, the set of natural cubic splines $S_{3,nat}(0,\ldots,l)$
is the set of cubic spline $s$ with $s''(0)=s''(l)=0$. Hence the dimension of  $S_{3,nat}(0,\ldots, l)$ is equal to $l+1$.
Now we are ready to define our parametrizations of $\Rb^{l+1}$.

There exist for each $l$ fixed a unique non symmetric matrix $\M^{(l)}$ and a unique symmetric matrix $\S^{(l)}$
such that
\be
\int_0^l|s(t)|^2dt=\{\s^{(l)}\}^\top\M^{(l)}\s^{(l)}=\{\s^{(l)}\}^\top\S^{(l)}\s^{(l)},
\ee
for all $s\in S_{3,nat}(0,\ldots, l)$.

We consider the following six parametrization matrices $\Theta^{(l)}=\M^{(l)}$, $\{\M^{(l)}\}^\top$, $\{\M^{(l)}\}^{-1}$, $\{\{\M^{(l)}\}^{-1}\}^\top$,
$\S^{(l)}$, $\{\S^{(l)}\}^{-1}$.
\begin{figure}[H]
    \begin{center}
        \begin{minipage}[t]{3cm}
            {\includegraphics[ width=3.6cm]{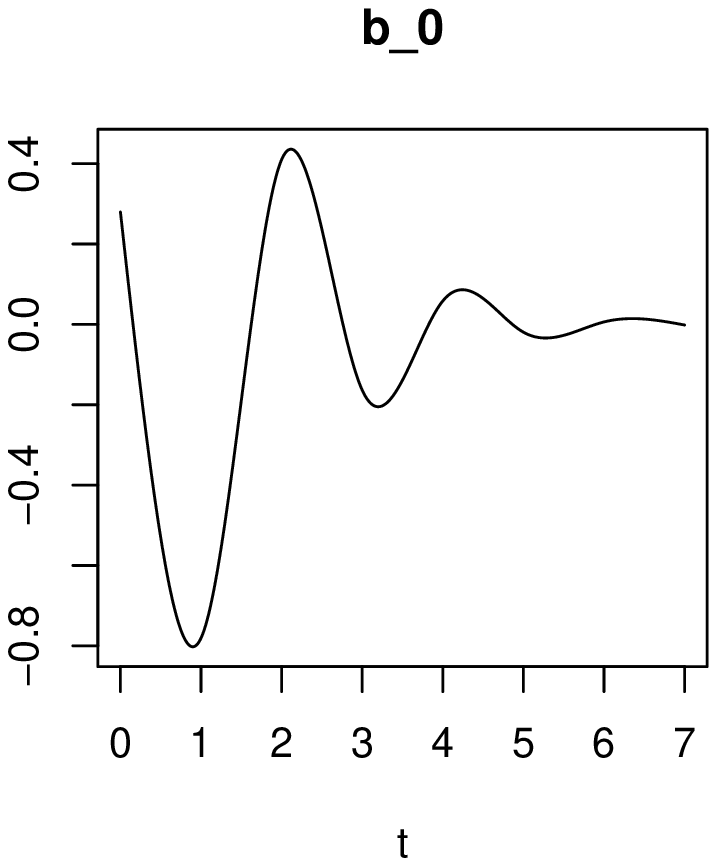}}
        \end{minipage}
        \begin{minipage}[t]{3cm}
            {\includegraphics[ width=3.6cm]{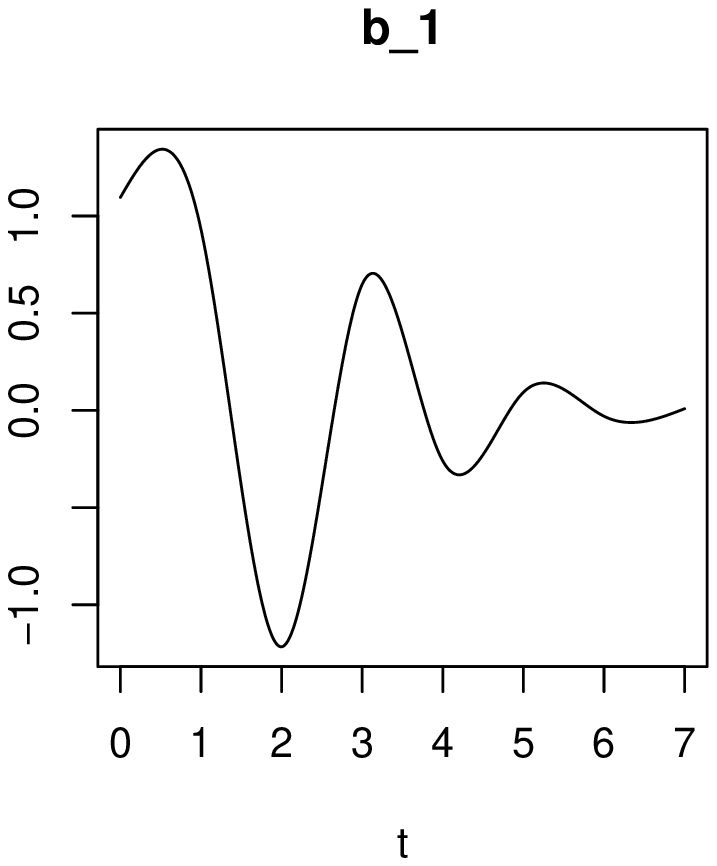}}
        \end{minipage}
        \begin{minipage}[t]{3cm}
            {\includegraphics[ width=3.6cm]{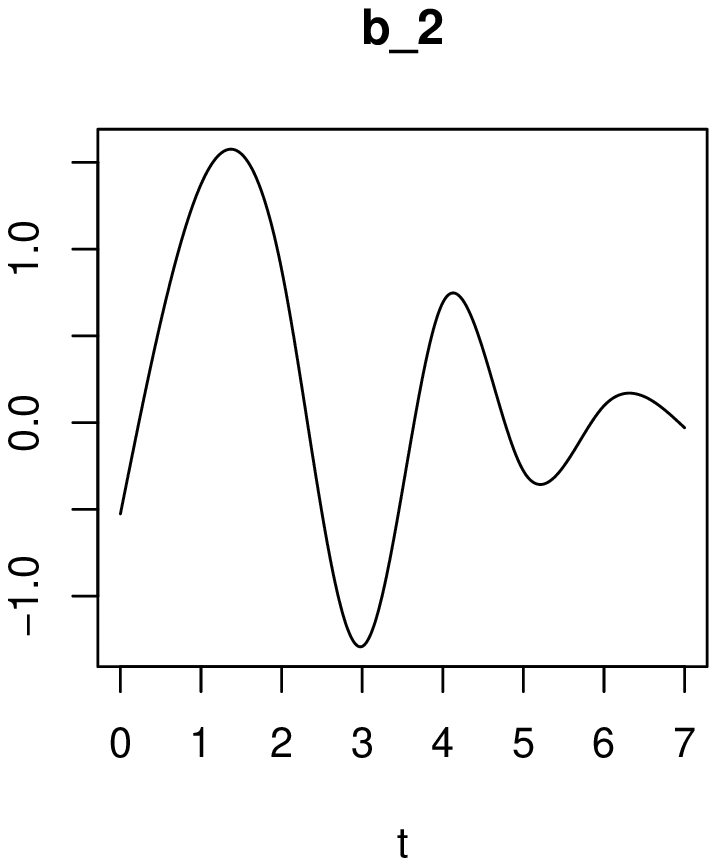}}
        \end{minipage}
        \begin{minipage}[t]{3cm}
            {\includegraphics[ width=3.6cm]{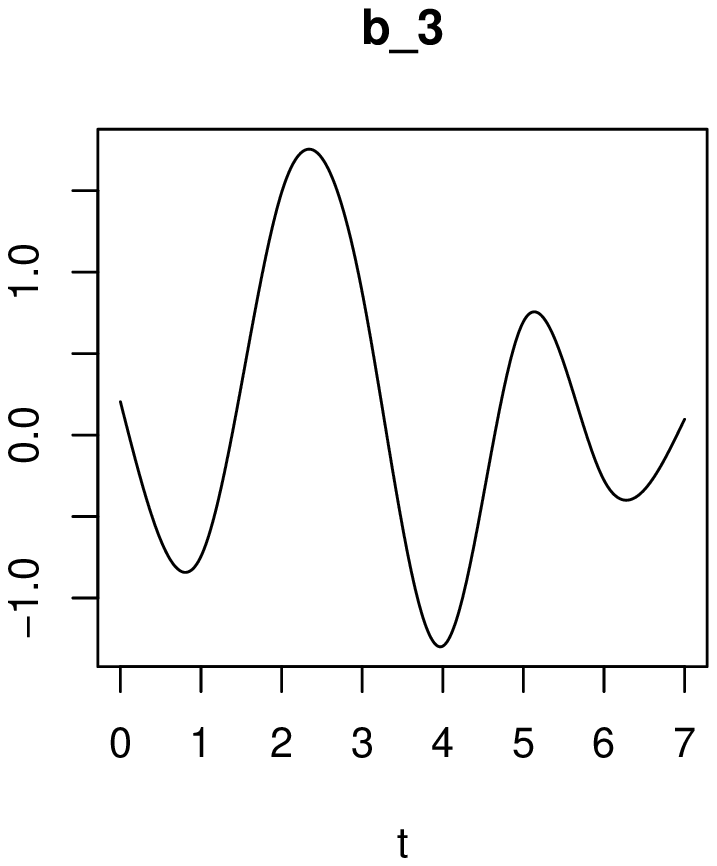}}
        \end{minipage}
        \begin{minipage}[t]{3cm}
            {\includegraphics[ width=3.6cm]{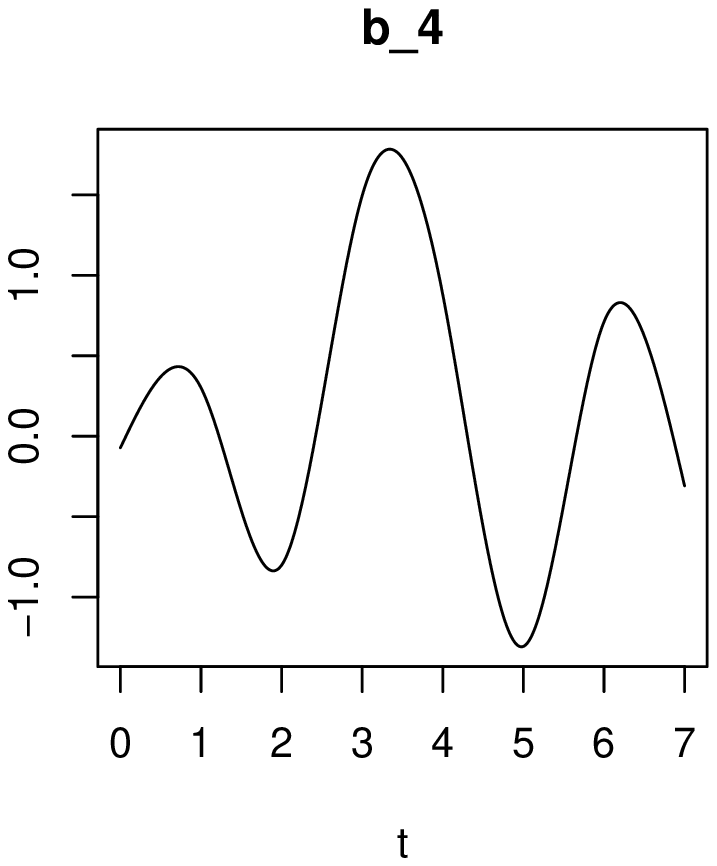}}
        \end{minipage}
        \begin{minipage}[t]{3cm}
            {\includegraphics[ width=3.6cm]{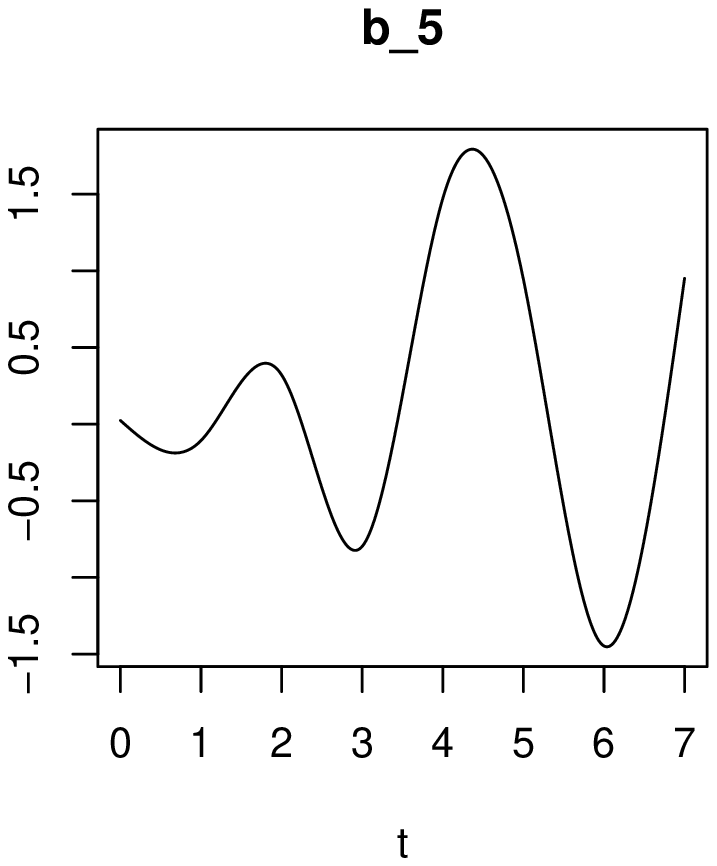}}
        \end{minipage}
        \begin{minipage}[t]{3cm}
            {\includegraphics[ width=3.6cm]{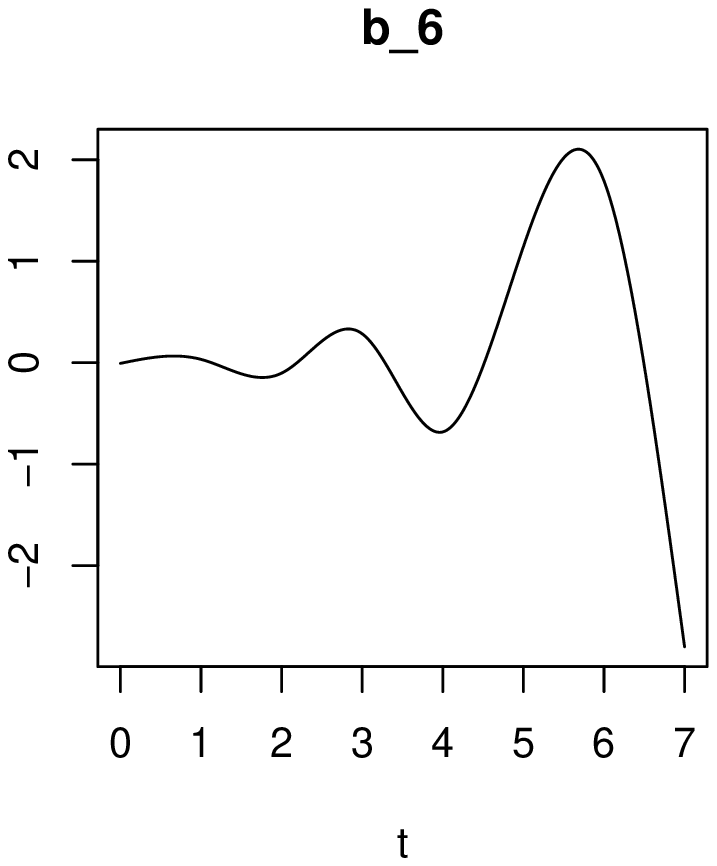}}
        \end{minipage}
        \begin{minipage}[t]{3cm}
            {\includegraphics[ width=3.6cm]{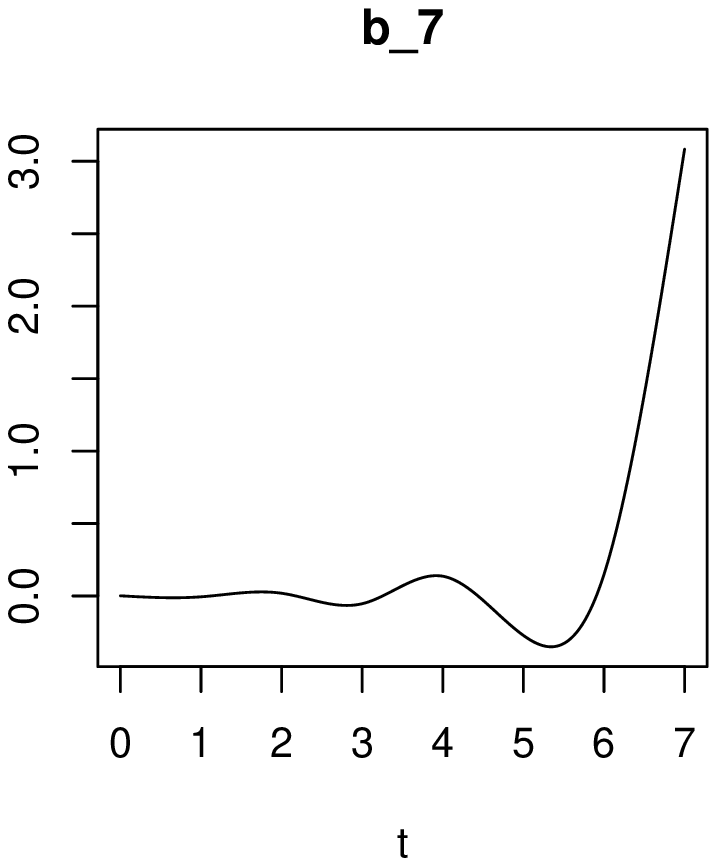}}
        \end{minipage}  \\
        \caption{ Represantation of the basis $\B^{(l)}=\M^{(l)}$ with $l=7$.}
    \end{center}
\end{figure}

\begin{figure}[H]
    \begin{center}
        \begin{minipage}[t]{3cm}
            {\includegraphics[ width=3.6cm]{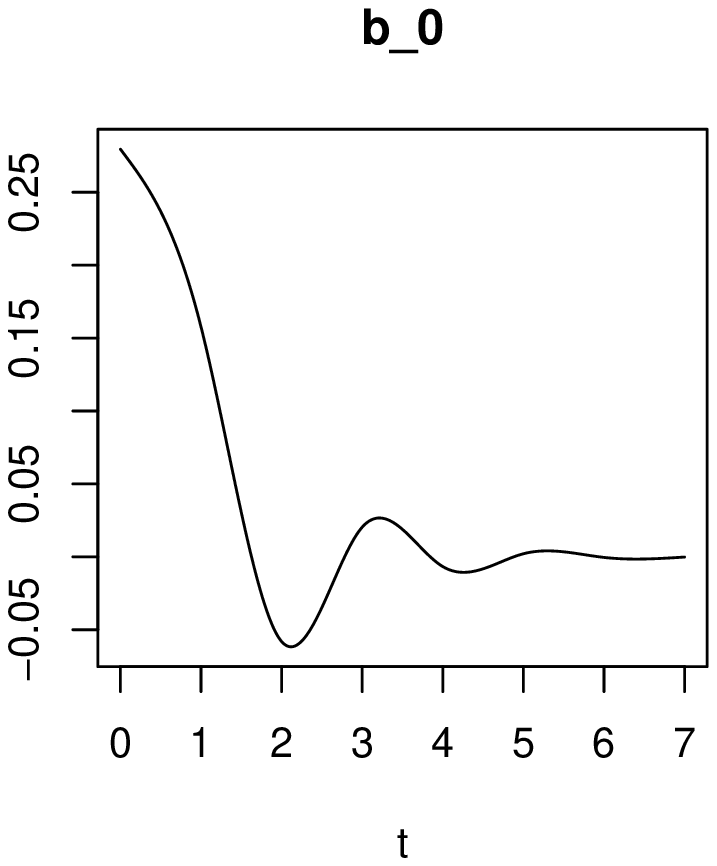}}
        \end{minipage}
        \begin{minipage}[t]{3cm}
            {\includegraphics[ width=3.6cm]{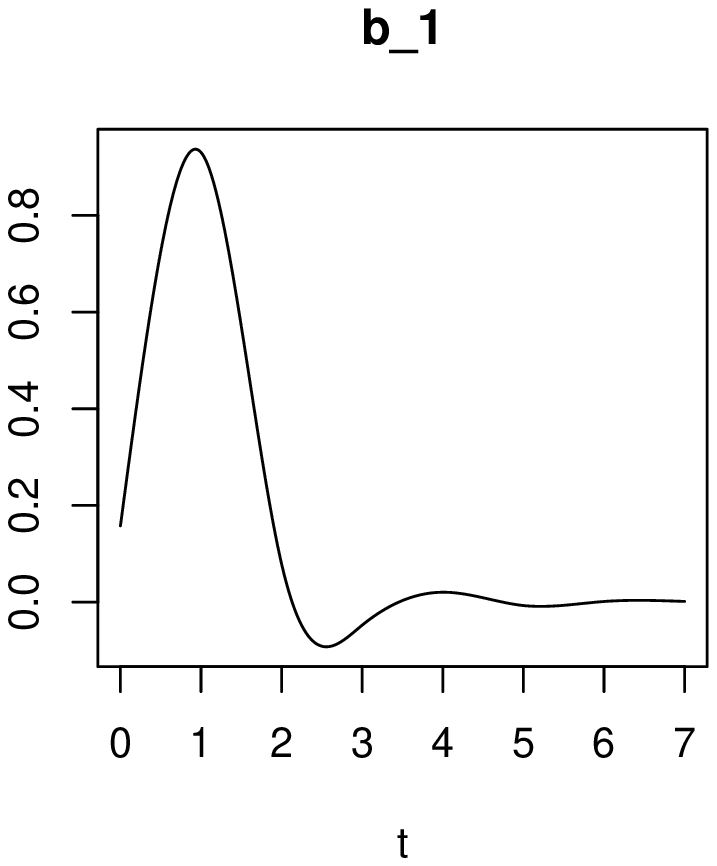}}
        \end{minipage}
        \begin{minipage}[t]{3cm}
            {\includegraphics[ width=3.6cm]{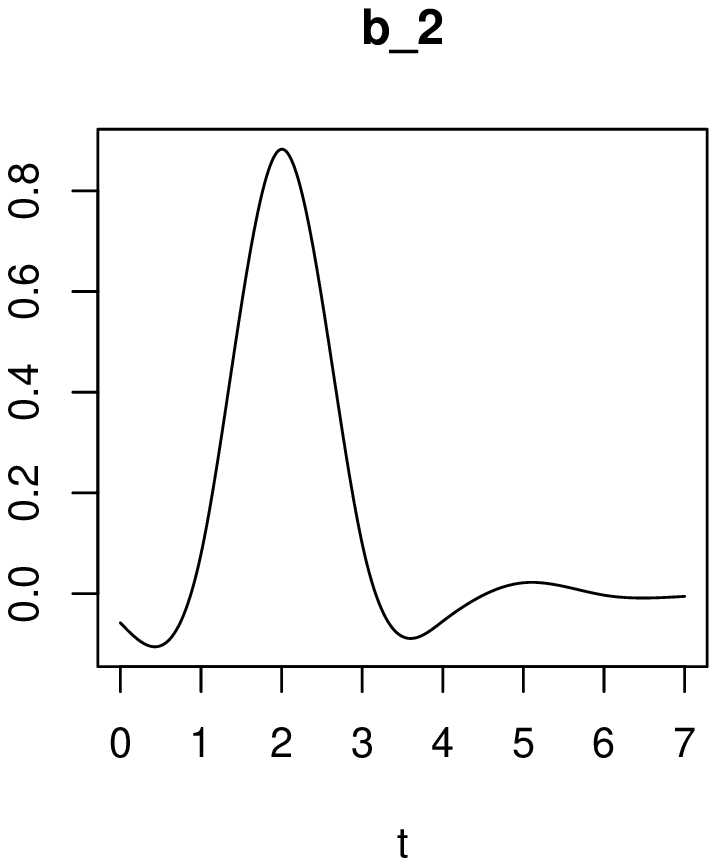}}
        \end{minipage}
        \begin{minipage}[t]{3cm}
            {\includegraphics[ width=3.6cm]{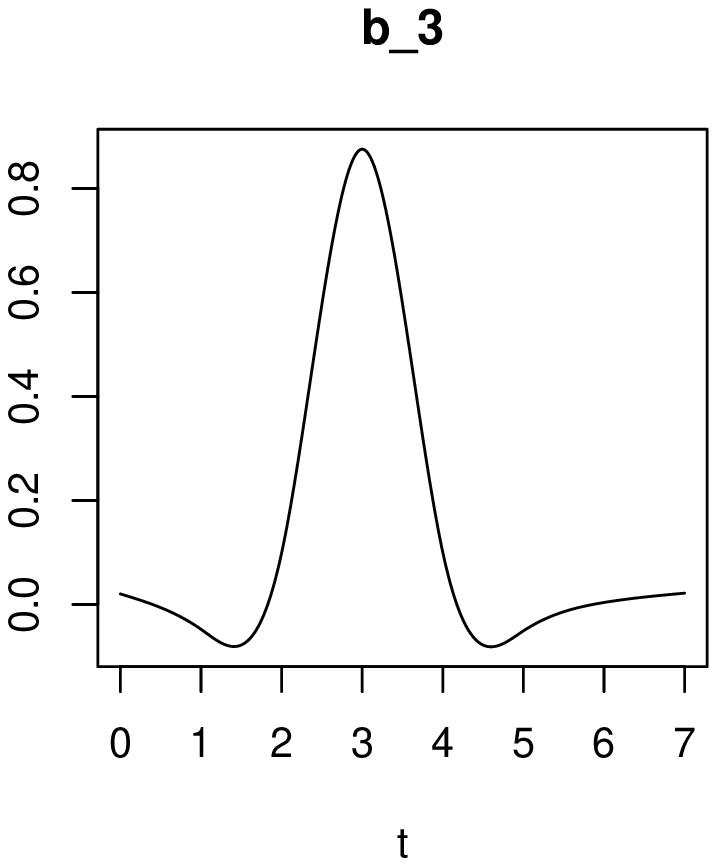}}
        \end{minipage}
        \begin{minipage}[t]{3cm}
            {\includegraphics[ width=3.6cm]{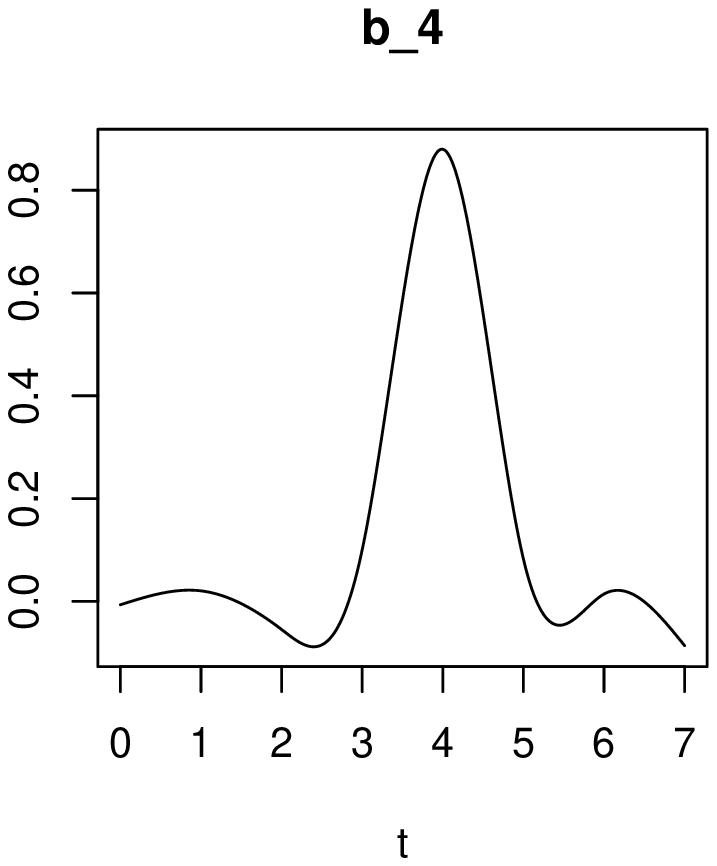}}
        \end{minipage}
        \begin{minipage}[t]{3cm}
            {\includegraphics[ width=3.6cm]{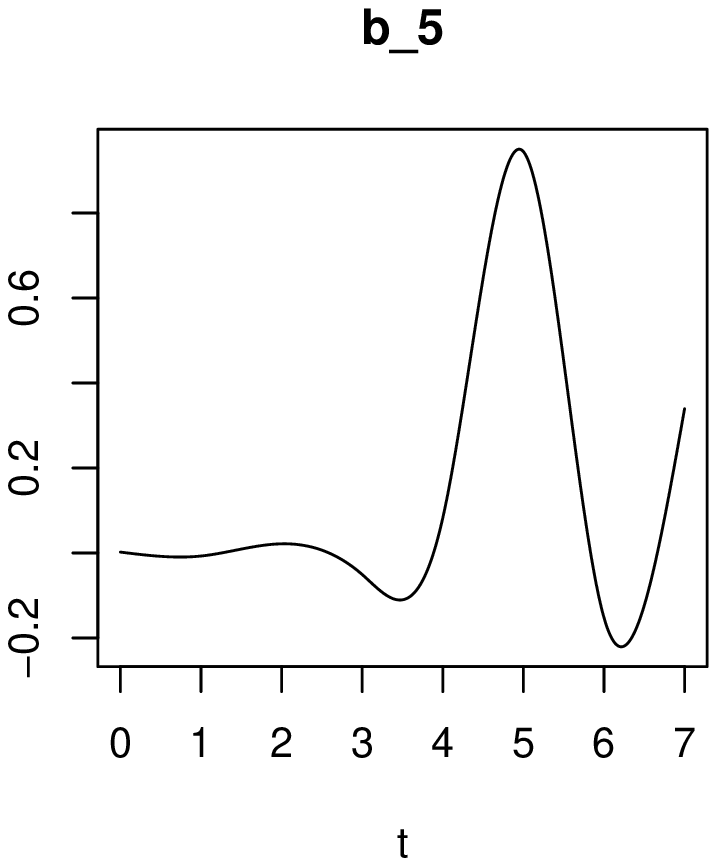}}
        \end{minipage}
        \begin{minipage}[t]{3cm}
            {\includegraphics[ width=3.6cm]{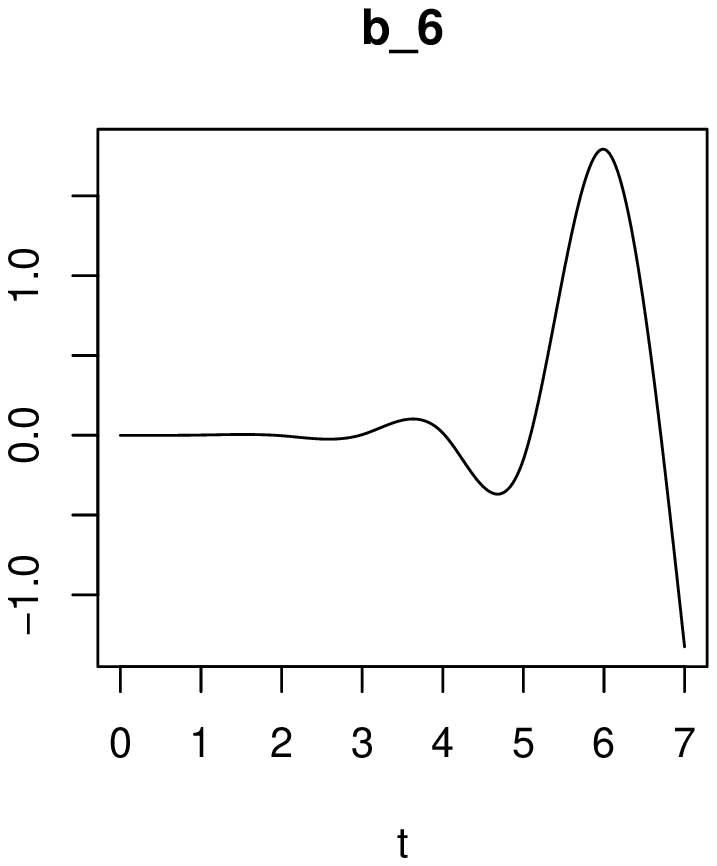}}
        \end{minipage}
        \begin{minipage}[t]{3cm}
            {\includegraphics[ width=3.6cm]{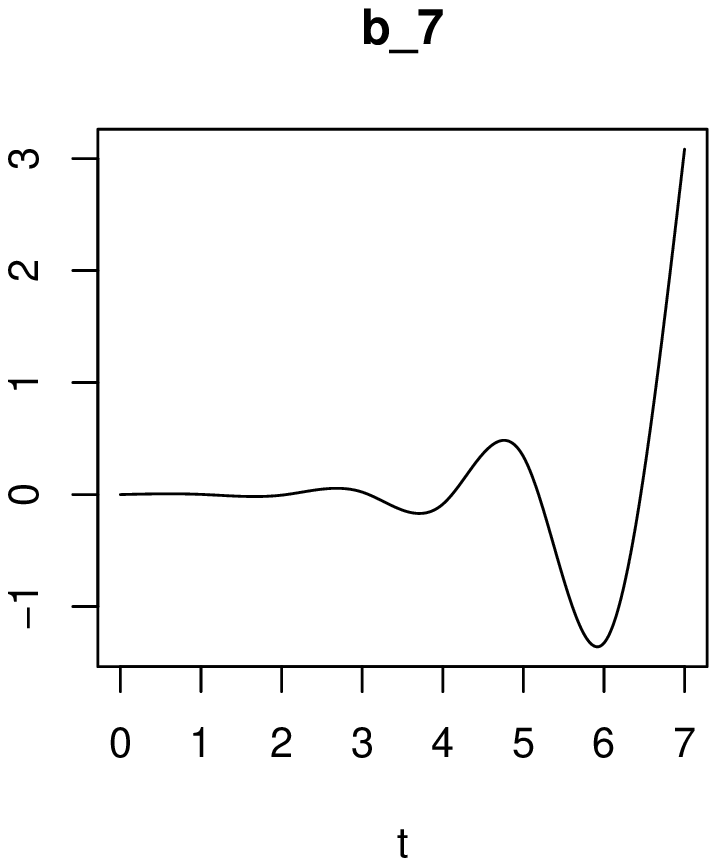}}
        \end{minipage}  \\
        \caption{Representation of the basis $\B^{(l)}=\S^{(l)}$ with $l=7$.}
    \end{center}
\end{figure}

\subsection{Real data application}
In the temperature prediction problem we are interested in the annual mean temperature observed in France and Morocco from 1901 to 2015. Data $\s^{(n)}=(s(0), \ldots, s(n))^\top$ with $n=114$ respectively for France and Morocco
are presented in Figure(\ref{data s}). Observe that $s(n)$ denotes the temperature of the year $n+1$.
\begin{figure}[H]
    \caption{Annual mean temperatures in France and Morocco from 1901 to 2015}
    \label{data s}
    \begin{center}
        \begin{minipage}[t]{10cm}
            {\includegraphics[ width=10cm]{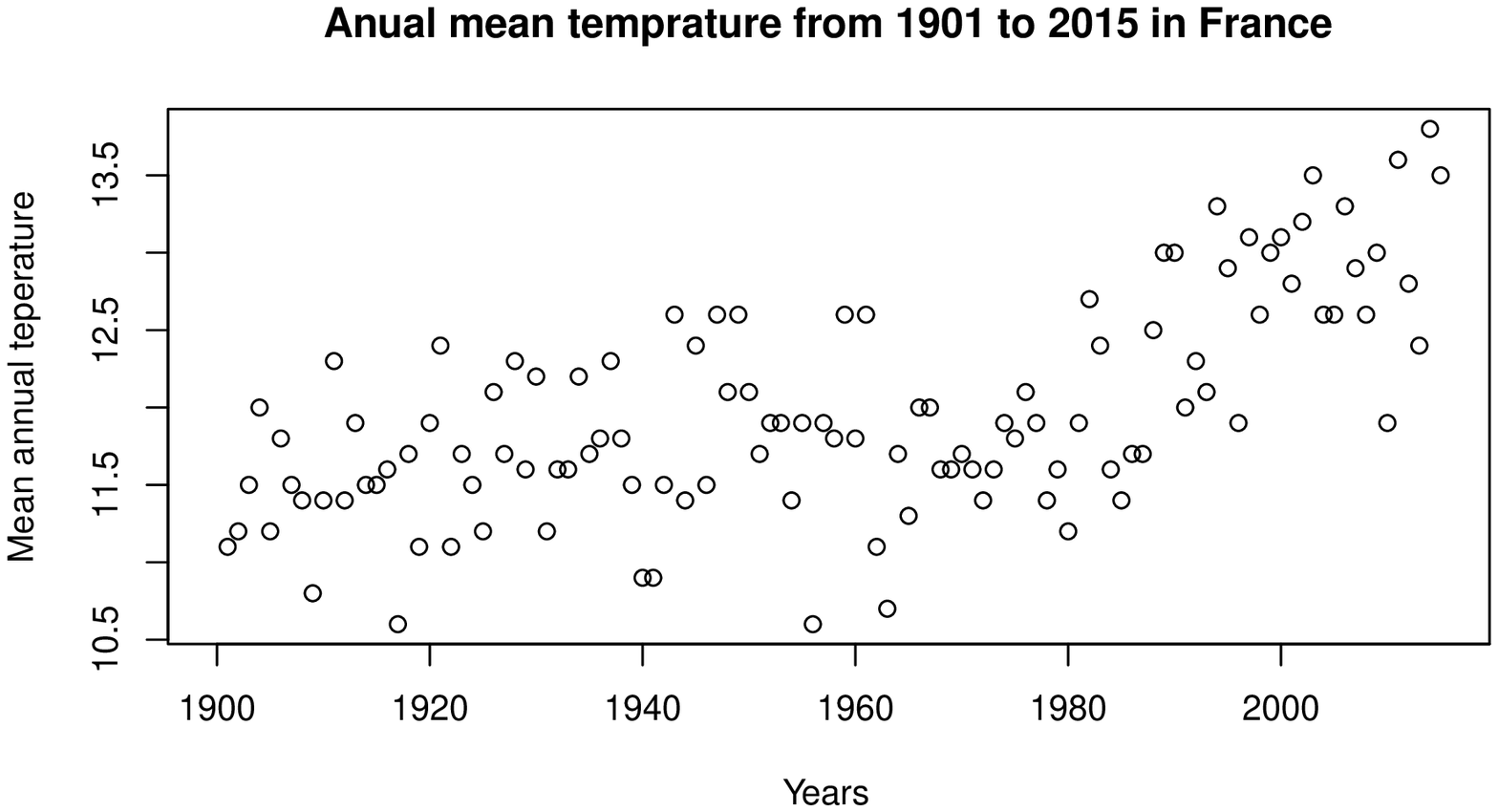}}
        \end{minipage}
        \begin{minipage}[t]{10cm}
            {\includegraphics[ width=10cm]{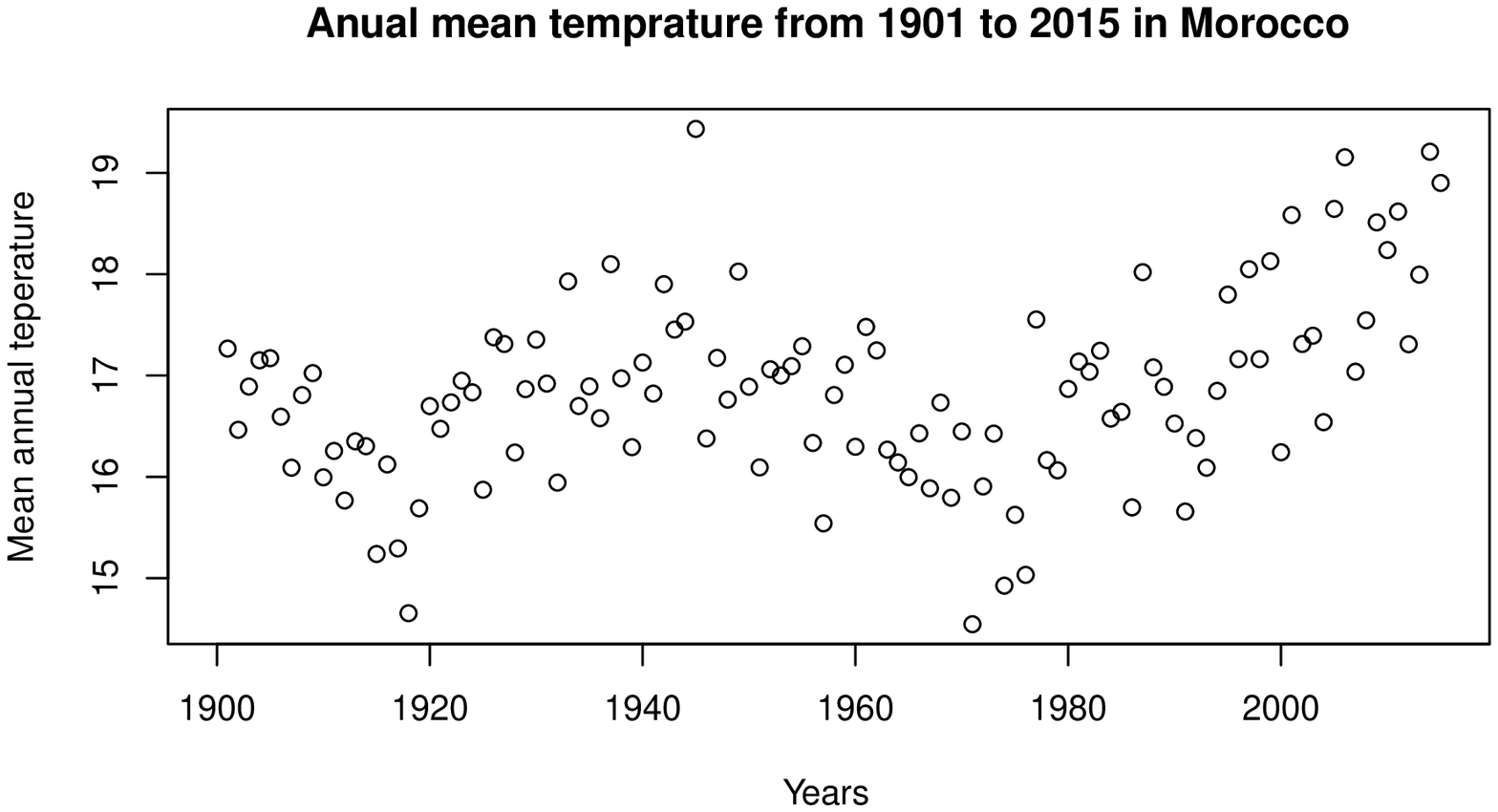}}
        \end{minipage}   \\

    \end{center}
\end{figure}
\subsubsection{Predictors}
\label{pre}
Our set of parametrizations $\mathcal{P}$ contains
\be
&&\M=(\M^{(l)}:\quad l=1, \hdots, n),\\
&&\M^\top=(\{\M^{(l)}\}^\top:\quad l=1, \hdots, n),\\
&&\M^{-1}=(\{\M^{(l)}\}^{-1}:\quad l=1, \hdots, n),\\
&&\{\M^{-1}\}^\top=(\{\{\M^{(l)}\}^{-1}\}^{\top}:\quad l=1, \hdots, n),\\
&&\S=(\S^{(l)}:\quad l=1, \hdots, n),\\
&&\S^{-1}=(\{\S^{(l)}\}^{-1}:\quad l=1, \hdots, n).
\ee
\\
Table (\ref{p-table}) shows  that for each $q=1,2,+\infty$ and the lag $L=4$ the optimal parametrization $\Theta(q)=\M^{-1}$ for both France and Morocco, but the optimal conservative rows $S_7(\Theta(q))$ do not coincide.
The optimal conservative rows $S_7(\Theta^{(114)}(q))$ are plotted in Figures (\ref{w2}) and (\ref{w1}).
The predictors of the temperature $s(114)$ (the temperature at the year 2015) and the true temperature is given in Table (\ref{pr-table}). The predictors of the temperature $s(115)$ (the temperature at the year 2016) is given in Table (\ref{pr-table2}). Splines of the true temperature and its optimal predictors are represented in Figure (\ref{s1}).
\begin{table}[H]
    \centering
    \caption{The optimal choice $\Theta(q)$ and $S_7(\Theta(q))$. }
    \label{p-table}
    \begin{tabular}{|l|l|l|l|}
        \cline{1-4}
        Country & \multicolumn{3}{c|}{\textbf{France}}\\ \cline{1-4}
        $q$ & 1 & 2 & $\infty$\\ \cline{1-4}
        $\Theta(q)$ & $\M^{-1}$ & $\M^{-1}$ & $\M^{-1}$ \\ \cline{1-4}
        $S_7(\Theta(q))$ & $S_{u\mbox{\small{tail2}}}$ $(u=86)$ & $S_{uv\small{fd}}$ $(u=93,v=5)$ & $S_{uv\small{fd}}$ $(u=81,v=8)$ \\ \cline{1-4}
        cost & 0.4233063 & 0.2784530 & 1.220770\\ \cline{1-4}
        Country & \multicolumn{3}{c|}{\textbf{Morocco}}\\ \cline{1-4}
        $q$ & 1 & 2 & $\infty$\\ \cline{1-4}
        $\Theta(q)$ & $\M^{-1}$ & $\M^{-1}$ & $\M^{-1}$ \\ \cline{1-4}
        $S_7(\Theta(q))$ & \multicolumn{2}{c|}{$S_{u\mbox{\small{tail1}}}$ $(u=73)$} & $S_{u\mbox{\small{tail2}}}$ $(u=41)$ \\ \cline{1-4}
        cost & 0.6183125 & 0.6027288 & 1.917094\\ \cline{1-4}
    \end{tabular}
\end{table}

\begin{table}[H]
    \centering
    \caption{The best prediction of  $s(114)$ using the optimal parametrization $\Theta(q)=\M^{-1}$ and
    the optimal conservative rows $S_7(\Theta(q))$.}
    \label{pr-table}
    \begin{tabular}{|l|l|l|l|}
        \cline{1-4}
        Country & \multicolumn{3}{c|}{\textbf{France}}\\ \cline{1-4}
        $q$ & 1 & 2 & $\infty$\\ \cline{1-4}
        True temperature &\multicolumn{3}{c|}{13.8} \\ \cline{1-4}
        Prediction & 13.03396 & 13.01986 & 12.86248 \\ \cline{1-4}
        Country & \multicolumn{3}{c|}{\textbf{Morocco}}\\ \cline{1-4}
        $q$ & 1 & 2 & $\infty$\\ \cline{1-4}
        True temperature &\multicolumn{3}{c|}{19.20845} \\ \cline{1-4}
        Prediction & 18.17489 & 18.17489 & 18.06860 \\ \cline{1-4}

    \end{tabular}
\end{table}

\begin{table}[H]
    \centering
    \caption{The best prediction of $s(115)$ (the temperature at the year 2016) using the optimal parametrization $\Theta(q)=\M^{-1}$ and the optimal conservative rows $S_7(\Theta(q))$.}
    \label{pr-table2}
    \begin{tabular}{|l|l|l|l|}
        \cline{1-4}
        Country & \multicolumn{3}{c|}{\textbf{France}}\\ \cline{1-4}
        $q$ & 1 & 2 & $\infty$\\ \cline{1-4}
        Prediction & 12.66792 & 13.06234 & 12.82844 \\ \cline{1-4}
        Country & \multicolumn{3}{c|}{\textbf{Morocco}}\\ \cline{1-4}
        $q$ & 1 & 2 & $\infty$\\ \cline{1-4}
        Prediction & 17.53958 & 17.53958  &17.76148 \\ \cline{1-4}
    \end{tabular}
\end{table}

\begin{figure}[H]
    \caption{The optimal conservative row $S_7(\Theta^{(114)}(q))$, $q=1,2,\infty$ for Morocco.}
    \label{w2}
    \begin{center}
        {\includegraphics[ width=12cm]{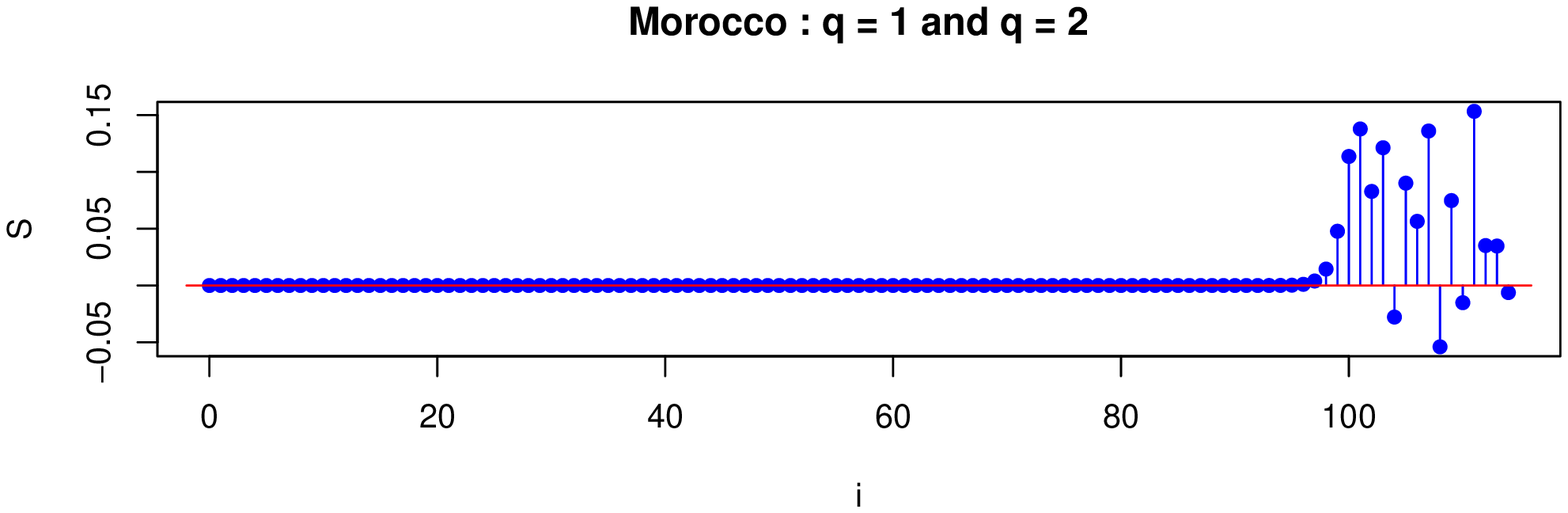}}\\
        {\includegraphics[ width=12cm]{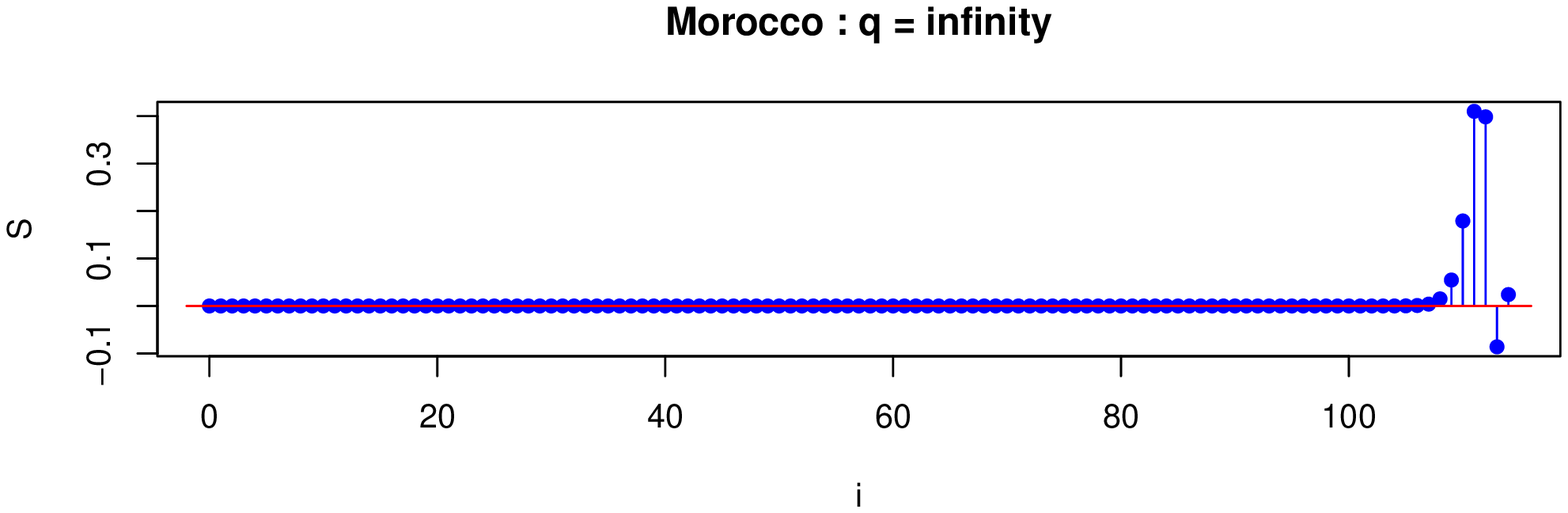}}\\

    \end{center}
\end{figure}

\begin{figure}[H]
    \caption{The optimal conservative row $S_7(\Theta^{(114)}(q))$, $q=1,2,\infty$ for France. }
    \label{w1}
    \begin{center}
            {\includegraphics[ width=12cm]{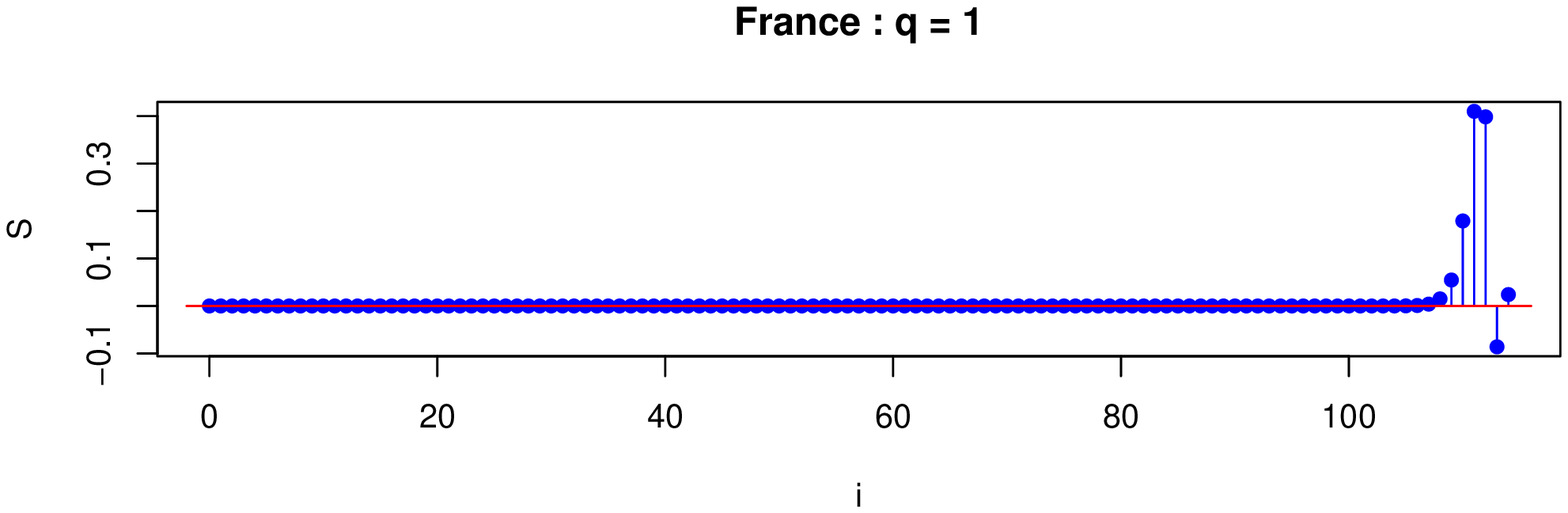}}\\
            {\includegraphics[ width=12cm]{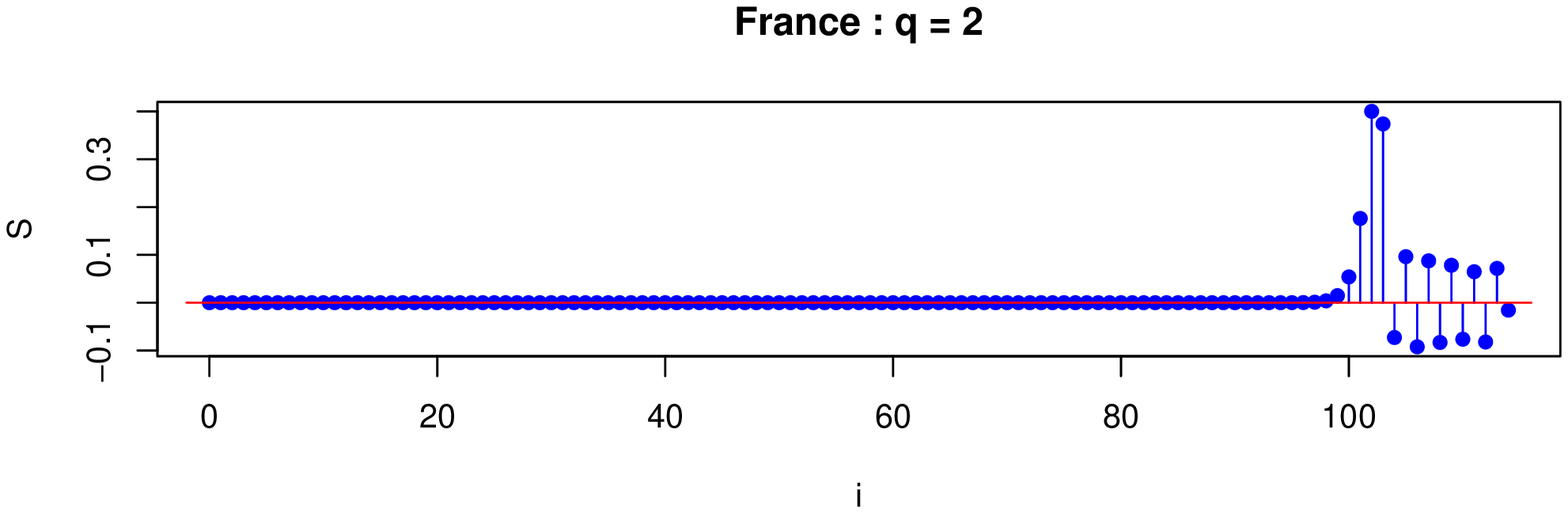}}\\
            {\includegraphics[ width=12cm]{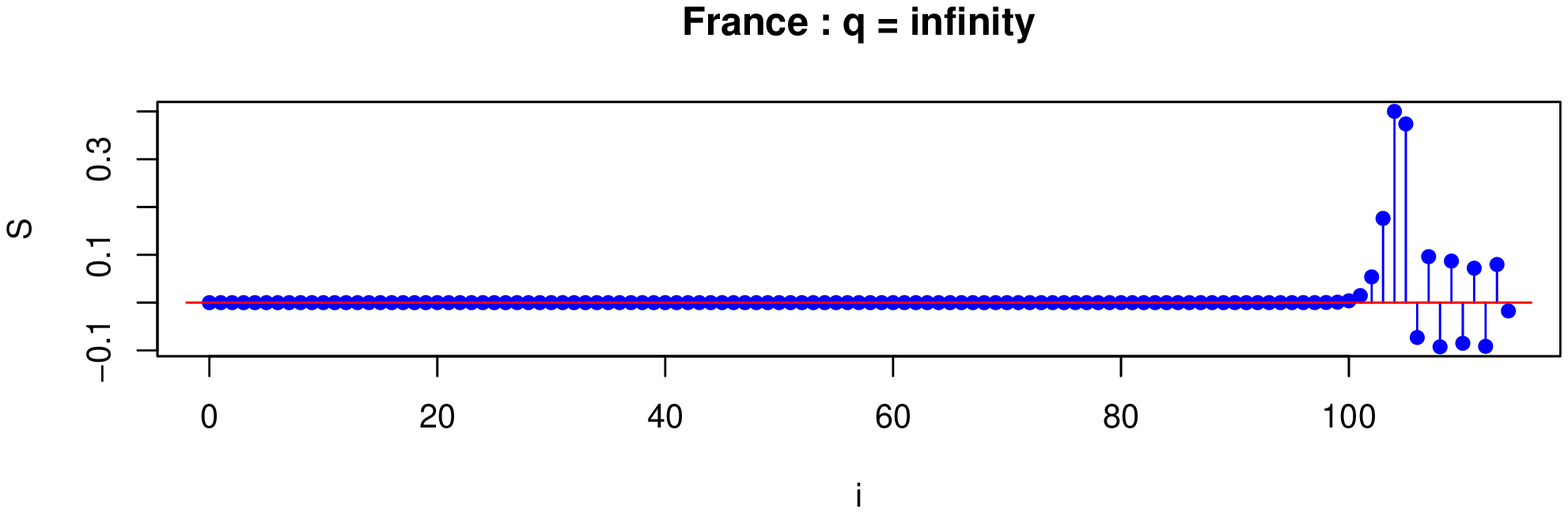}}\\

    \end{center}
\end{figure}

 \begin{figure}[H]
    \caption{The splines of $\s^{(n)}$ and its optimal predictors, with $n=114$.}
    \label{s1}
    \begin{center}
        {\includegraphics[ width=12cm]{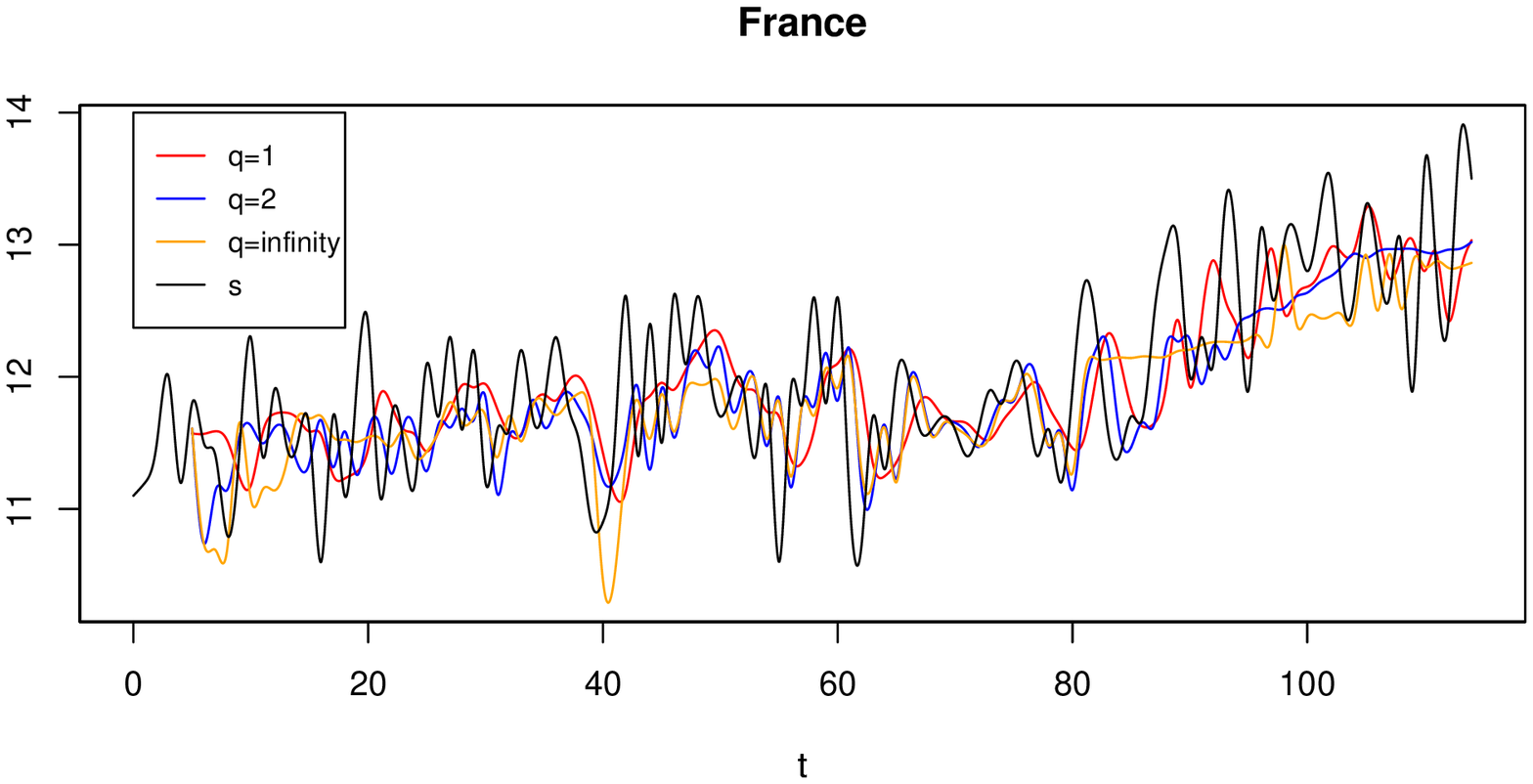}}\\
        {\includegraphics[ width=12cm]{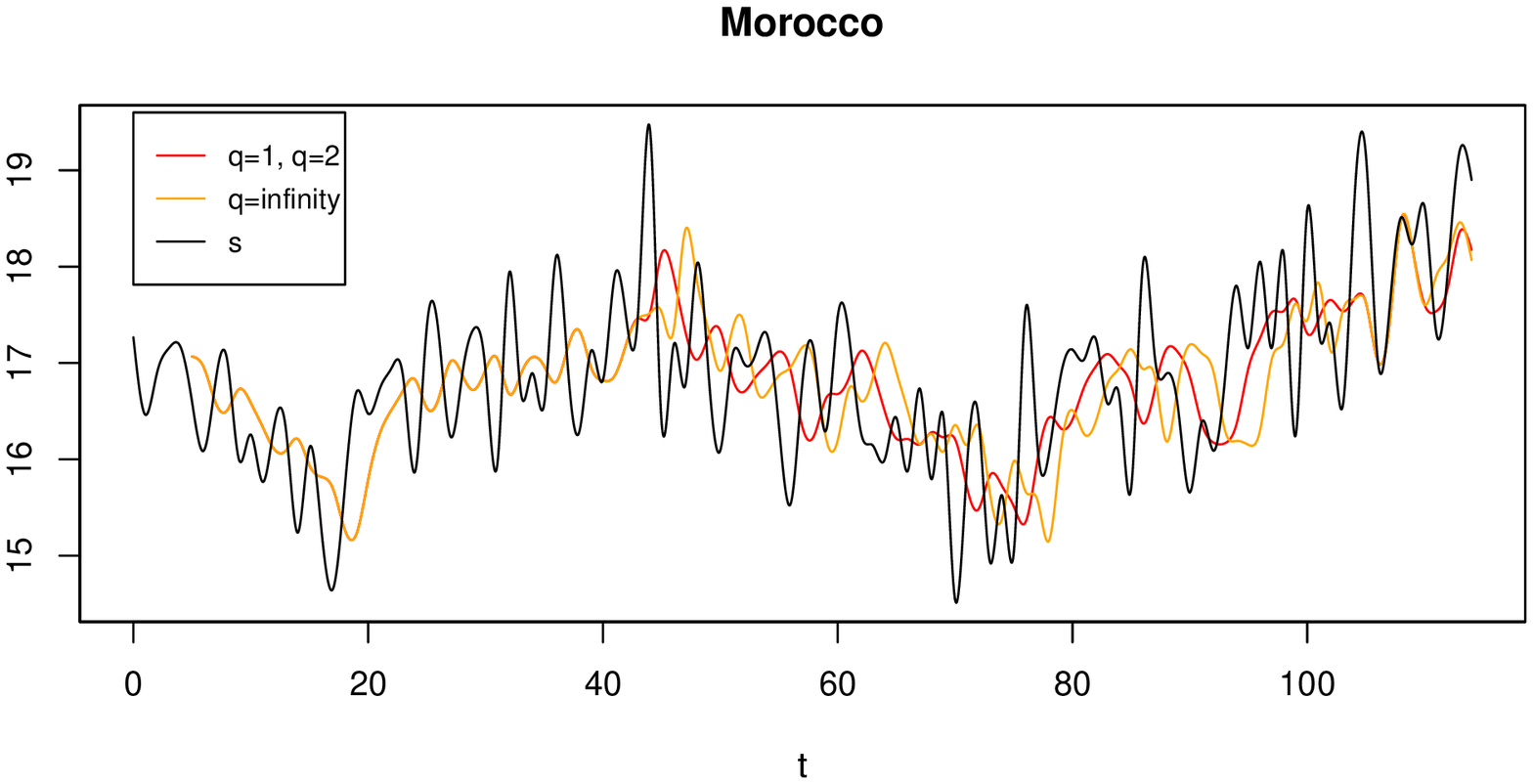}}\\

    \end{center}
 \end{figure}

 {\bf Conclusion.} Having a time series $s(0)$, $\hdots$, $s(n)$
with values in $\Rb$, we showed how to predict the value $s(n+1)$
from each parametrization of the set $\Rb^{n+1}$. We also provided optimality criteria to select the best
predictor. This work can be extended to time series $s(i)\in K$ with $K$ is any field or vector space.

\end{document}